\documentclass[12pt, draftclsnofoot, journal, letterpaper, onecolumn]{IEEEtran}

\usepackage{amsfonts}
\usepackage[dvips]{graphicx}
\usepackage{times}
\usepackage{cite}
\usepackage{amsmath}
\usepackage{array}
\usepackage{amssymb}

\newcommand{\bs}{\boldsymbol}

\usepackage{stfloats}
\usepackage{slashbox}
\usepackage{graphicx}
\usepackage{footnote}
\usepackage{booktabs}
\usepackage{array}
\usepackage{algorithmic}
\usepackage{algorithm}
\usepackage{subeqnarray}
\usepackage{cases}
\usepackage{threeparttable}
\usepackage{color}
\usepackage{hyperref}
\usepackage[normalsize]{subfigure}

\newtheorem{remark}{Remark}
\newtheorem{proposition}{Proposition}

\begin{document}

\title{QoS-Aware Transmission Policies for OFDM Bidirectional Decode-and-Forward Relaying}

\author{Yuan Liu, \IEEEmembership{Member, IEEE}, Jianhua Mo, and Meixia Tao, \IEEEmembership{Senior Member, IEEE}
\thanks{Manuscript received May 19, 2012; revised September 29, 2012 and December 14, 2012; accepted January 22, 2013. The Editor coordinating the review of this paper and approving it for publication was Prof. Luca Sanguinetti.}
\thanks{The authors are with the Department of
Electronic Engineering at Shanghai Jiao Tong University, Shanghai,
200240, P. R. China. Email: \{yuanliu, mjh, mxtao\}@sjtu.edu.cn.}
\thanks{This work is supported by the Program for New Century Excellent Talents in University (NCET) under grant NCET-11-0331 and the National 973 project under grant 2012CB316100. This paper was presented in part at IEEE Global Telecommunications Conference (GLOBECOM), Anaheim, California, USA, December 3-7, 2012 \cite{YuanGC12}.}
}

\maketitle

\vspace{-1.5cm}
\begin{abstract}
Two-way relaying can considerably improve spectral efficiency in
relay-assisted bidirectional communications. However, the benefits
and flexible structure of orthogonal frequency division multiplexing
(OFDM)-based two-way decode-and-forward (DF) relay systems is much
less exploited. Moreover, most of existing works have not considered
quality-of-service (QoS) provisioning for two-way relaying. In this
paper, we consider the OFDM-based bidirectional transmission where a pair of users
exchange information with or without the assistance of a single
DF relay. Each user can communicate with the
other via three transmission modes: direct transmission, one-way
relaying, and two-way relaying. We jointly optimize the transmission
policies, including power allocation, transmission mode selection,
and subcarrier assignment for maximizing the weighted sum rates of the
two users with diverse quality-of-service (QoS) guarantees. We formulate the
joint optimization problem as a mixed integer programming problem.
By using the dual method, we efficiently solve the problem in an asymptotically optimal manner. Moreover, we derive the capacity
region of two-way DF relaying in parallel channels. Simulation results
show that the proposed resource-allocation scheme can substantially
improve system performance compared with the conventional schemes. A number of interesting insights are also provided via comprehensive simulations.

\end{abstract}


\section{Introduction}

Orthogonal frequency division multiplexing (OFDM) is a leading physical layer transmission technique for high spectral efficiency and date rate in broadband wireless communication systems. OFDM also naturally provides a multiple-access method, as known as OFDMA by allocating different subcarriers to different users in multiuser environments \cite{Seong,Kwak,Zhu1,Zhu2}.
On the other hand, cooperative relay has received much interests due to its capabilities of improving system performance, such as throughput enhancement, power saving, and communication coverage extension \cite{Sendonaris1,Sendonaris2,Laneman}.
Combining relaying architecture with OFDM transmission is a powerful technique for
broadband wireless communication, and thus adopted in many current
and next generation standards, i.e., 3GPP Long Term Evolution
Advanced (LTE-Advanced) and IEEE 802.16m.

However, the traditional one-way relaying is less spectrally
efficient due to the practical half-duplex constraint. To overcome
this problem, two-way relaying has been recently proposed
\cite{Katti,Shengli,Rankov,Popovski,Kim}. Its principle is to apply
network coding at the relay node to mix the signals received from
two links for subsequent forwarding, and then apply
self-interference cancelation at each destination to extract the
desired signals. Naturally, it is promising and attractive to
exploit network coding gain by dynamic resource allocations for
improving spectral efficiency in OFDM bidirectional relay systems.

There are several works on resource allocation in OFDM bidirectional
relay systems
\cite{Jitvanichphaibool,Jang,Dong,Ho,YuanTWC10,YuanTCOM12}.
These works can be divided
into two categories: \emph{per-subcarrier basis}
\cite{Jitvanichphaibool,Jang,Dong} and \emph{subcarrier pairing basis}
\cite{Ho,YuanTWC10,YuanTCOM12}. The first category assumes that the
two-hop cooperative transmission, i.e., source-to-relay link and
relay-to-destination link use the same subcarrier. Such
per-subcarrier basis significantly simplifies the optimization
problems but does not fully utilize the channel dynamics. For
instance, the authors in \cite{Jitvanichphaibool} studied power and
subcarrier allocation for OFDM two-way relaying with both
amplify-and-forward (AF) and decode-and-forward (DF) strategies. By
using dual decomposition method, the problem was decomposed into
per-subcarrier subproblems that can be solved independently. A
two-step suboptimal method for power allocation for OFDM two-way AF
relaying was proposed in \cite{Jang}, where power is first allocated
in each subcarrier for a given per-subcarrier power constraint, then
the per-subcarrier power constraints are coordinated to satisfy a
total peak power constraint of the system. The authors in \cite{Dong} showed that
the optimal power allocation for OFDM two-way AF relaying with a
total peak power constraint turns out to be a two-step approach as
in \cite{Jang}.

Different from the per-subcarrier basis, the subcarrier pairing
basis allows the subcarriers in the first and second hops to be
paired and then a better performance can be provided \cite{Sanguinetti}. In \cite{Ho},
power was first allocated by water-filling and then subcarriers were
paired by a greedy heuristic method for OFDM two-way AF relaying. In
\cite{YuanTWC10}, the authors investigated the subcarrier pairing based joint optimization problem of transmission mode selection, subcarrier assignment, and relay
selection for OFDMA bidirectional relay cellular networks by an ant
colony optimization method from a graph theoretical perspective. In
\cite{YuanTCOM12}, the authors studied optimal subcarrier
and relay assignment for OFDM two-way relay systems  using a bipartite graph matching algorithm.

In view of these existing works, our paper is motivated in
threefold: Firstly, both per-subcarrier basis and
subcarrier pairing basis are \emph{not} optimal for two-way DF
relaying, where the information from one set of subcarriers in the
first hop can be decoded and re-encoded jointly and then transmitted
over a different set of subcarriers in the next hop. This is referred as \emph{subcarrier set basis} in this paper. Secondly,  by use of the parallel OFDM relaying
architecture, the bidirectional communication can be completed by
three transmission modes, namely direct transmission, one-way relaying, and
two-way relaying. Moreover, power allocation, subcarrier assignment,
and transmission mode selection are tightly coupled with each other.
How to jointly coordinate these transmission policies and how much power and spectral efficiencies are contributed by different transmission modes, are crucial
but more importantly, have not been considered for OFDM
bidirectional relay systems. Thirdly, one challenging issue to be addressed for future
developments of wireless networks is how to meet user's diverse
quality-of-service (QoS) requirements.
Real-time applications, such as voice transmission and video streaming, are highly delay-sensitive and need reliable QoS guarantees.
Therefore, it is of great importance to study dynamic resource allocation schemes for supporting diverse QoS requirements.
Nevertheless, what the impacts of resource allocation on QoS guarantees for OFDM
bidirectional relay systems, has also not been addressed in the
literature.

In this paper, we consider the above three issues in a classical OFDM two-way relaying scenario, where a pair of users
exchange information with assistance of a single DF relay
using OFDM. We enable each user to communicate with the other via three transmission modes simultaneously but over different sets of subcarriers. It is worth mentioning that, to our best knowledge, such a hybrid bidirectional transmission was only investigated in our previous work \cite{YuanTWC10}. However, \cite{YuanTWC10} is based on subcarrier pairing basis and does not consider power allocation and QoS guarantees. The main differences between this paper and the related works \cite{Jitvanichphaibool,Jang,Dong,Ho,YuanTWC10,YuanTCOM12} are stated in Table I.

The main contributions and results of this paper are summarized as
follows:
\begin{itemize}
   \item We formulate a joint optimization problem of power allocation, subcarrier assignment, and transmission mode selection for OFDM bidirectional DF relaying. The previous works often consider partial resources of this problem. Our objective is to maximize the weighted sum rates of the two users with diverse QoS guarantees. The joint problem is a mixed integer programming problem and  NP-hard. By using the dual method, we develop an asymptotically optimal algorithm to find the QoS-aware transmission policies with linear complexity of the number of subcarriers. Moreover, we derive the achievable capacity region of two-way DF relaying in parallel relay channels.

    \item Simulation results reveal that for the OFDM two-way DF relaying, the proposed subcarrier set relaying basis can achieve substantial throughput gain over the conventional subcarrier pairing relaying basis. For the OFDM bidirectional with hybrid transmission modes, the importance of one-way relaying is decreasing as signal-to-noise ratio (SNR) increases. On the contrary, the importance of direct transmission and two-way relaying are increasing with SNR, and two-way relaying dominates the system performance. We find that for a given user with more stringent rate QoS requirement, one-way relaying devotes more throughput and direct transmission devotes less. Moreover, for any rate QoS requirement, two-way relaying always dominates the system throughput. We also show that direct transmission dominates the system performance when the relay node is closer either of the two users, and one- and two-way relaying work well when the relay node locates at the midpoint of the two users.
\end{itemize}

The remainder of this paper is organized as follows. Section II describes the system model
and presents the rigorous problem formulation. The proposed
dual-based resource-allocation algorithm is detailed in Section III. Comprehensive simulation results
are illustrated in Section VI. Finally, we conclude this paper in Section V.

\section{Optimization Framework}

\subsection{System Model}
We consider the relay-assisted bidirectional communication as shown in Fig.~\ref{fig:system}, which consists of a pair of users $A$ and $B$, and a single relay $R$. Each user can communicate with the other directly or through the relay. Thus, each user can communicate with the other via three transmission modes, namely, direct transmission, one- and two-way relaying. In this paper, the
two-phase two-way relaying protocol is applied, i.e., the first
phase is multiple-access (MAC) phase and the second phase is the
broadcast (BC) phase \cite{Rankov,Popovski,Kim}.
Each node can transmit and receive at the same time but on different
frequencies. For both one-way relaying and two-way relaying, the
relay adopts DF  strategy and the delay between the first and second
hops can be negligible compared with the duration of a transmission
frame.
For example, Fig.~\ref{fig:system} shows that $A$ and $B$ can use subcarrier $\{9\}$
for the MAC phase but the relay can use subcarriers $\{7, 8\}$ in
the BC phase. Notice that such subcarrier set basis relaying is
also applicable for one-way relaying.

\subsection{Channel Model}

The wireless channels are modeled by large-scale path loss,
shadowing, and small-scale frequency-selective Rayleigh fading. It
is assumed that the transmission to both users is divided into
consecutive frames, and the fading remains unchanged within each
transmission frame but varies from one frame to another. We also
assume that channel estimation is perfectly known at all nodes. Note
that in relay-assisted systems such as IEEE 802.16m, relay nodes are
usually fixed. Such that the task of centralized resource allocation
can be embedded at the relay.  Without loss of generality, the
additive white noises at all nodes are assumed to be independent
circular symmetric complex Gaussian random variables, each having
zero mean and unit variance. The channel coefficients from node $j$
to node $j'$ on subcarrier $n$ are denoted as $h_{j,j',n}$, where
$j,j'\in\{A,B,R\}$, $j\neq j'$.

\subsection{Problem Formulation}
We use superscripts $a$, $b$, and $c$ to denote direct transmission,
one-way relaying, and two-way relaying, respectively.  We first
introduce the following three sets of binary assignment variables with respect to the three transmission modes:
\begin{itemize}
\item[-] $\rho_{k,n}^a$ indicates whether subcarrier $n$ is
assigned to user $k$ for direct transmission, $k\in\{A,B\}$.
\item[-] $\rho_{k,n,i}^b$ indicates whether subcarrier $n$ is
assigned to user $k$ at the $i$-th hop of one-way relaying, $k\in\{A,B\}$, $i=1,2$.
\item[-] $\rho_{n,i}^c$ indicates whether subcarrier $n$ is
assigned to the user pair at the $i$-th hop of two-way relaying, $i=1,2$.
\end{itemize}
As mentioned in previous works \cite{YuanTWC10,YuanTCOM12,YuanJSAC12}, the bidirectional links must occur in pair for two-way relaying. Therefore, in our case, the user index $k$ is not involved in $\rho_{n,i}^c$. In order to avoid interference, these binary variables must satisfy the following constraint:
\begin{equation}\label{eqn:rho}
  \sum_{k\in\{A,B\}}\rho_{k,n}^a + \sum_{k\in\{A,B\}}\sum_{i=1}^2\rho_{k,n,i}^b + \sum_{i=1}^2\rho_{n,i}^c\leq 1,~\forall n\in\mathcal {N},
\end{equation}
where $\mathcal N=\{1,\cdots,N\}$ is the set of subcarriers.

Let $p_{k,n}^a$ denote the transmit power of user $k$ over subcarrier $n$ for direct transmission, $p_{k,R,n}^b$ and $p_{k,R,n}^c$ as the transmit power of user $k$ to the relay over subcarrier $n$ for one- and two-way relaying, respectively, $k\in\{A,B\}$. Let $P_k$ be the total power of user $k$, then the power
allocation policy of user $k$ should satisfy:
\begin{equation}\label{eqn:pk_peak}
  \sum_{n=1}^N\left( p_{k,n}^a + p_{k,R,n}^b + p_{k,R,n}^c \right)\leq
  P_k,~k\in\{A,B\}.
\end{equation}

Denote $p_{R,k,n}^b$ as the transmit power from relay node to user $k$ over subcarrier $n$ using one-way relaying, $k\in\{A,B\}$. Denote $p_{R,n}^c$ as the transmit power of the relay node over subcarrier $n$ for two-way relaying. The relay node is subject to the peak power constraint $P_R$, which can be expressed as
\begin{equation}\label{eqn:pr_peak}
  \sum_{n=1}^N\left( \sum_{k\in\{A,B\}}p_{R,k,n}^b + p_{R,n}^c \right)\leq
  P_R.
\end{equation}

After introducing the assignment and power variables, now we briefly present the achievable rates for the three transmission modes.

For direct transmission mode, the achievable rate of user $k$ over subcarrier $n$ can be easily given by
\begin{equation}
  R_{k,n}^a=C(p_{k,n}^a|h_{k,k',n}|^2),~k,k'\in\{A,B\},k\neq k',
\end{equation}
where $C(x)=\log_2(1+x)$. Then the achievable rate of user $k$ by using direct transmission mode is
\begin{equation}
  R_k^a=\sum_{n=1}^N \rho_{k,n}^a R_{k,n}^a,~k\in\{A,B\}.
\end{equation}

For one-way relaying transmission mode, the achievable rates of the first and second hops for user $k$ can be respectively written as:
\begin{eqnarray}
  &&R_{k,n,1}^b=C(p_{k,R,n}^b|h_{k,R,n}|^2),\nonumber\\
  &&R_{k,n,2}^b=C(p_{R,k',n}^b|h_{R,k',n}|^2),
\end{eqnarray}
with $k,k'\in\{A,B\},k\neq k'$. The end-to-end achievable rate of user $k$ by using one-way relaying is the minimum of the rates achieved in the two hops, which can be expressed as
\begin{equation}
  R_{k}^b = \min\left\{\sum_{n=1}^N\rho_{k,n,1}^b R_{k,n,1}^b, \sum_{n=1}^N\rho_{k,n,2}^b R_{k,n,2}^b\right\},~k\in\{A,B\}.
\end{equation}

For the two-way DF relaying, prior work has studied the capacity region for single-channel case \cite{Rankov,Popovski,Kim}. Based on
these results, we derive the capacity region of OFDM
two-way DF relaying by the following proposition.
\begin{proposition}
  The capacity region $(R_A^c,R_B^c)$ of OFDM two-way DF relaying is given by
  \begin{equation}\label{eqn:capacity}
    \begin{split}
      R_A^c&\leq\sum_{n=1}^N\rho_{n,1}^c R_{A,n,1}^c \\
      R_B^c&\leq\sum_{n=1}^N\rho_{n,1}^c R_{B,n,1}^c \\
      R_A^c+R_B^c&\leq\sum_{n=1}^N\rho_{n,1}^c R_{AB,n,1}^c \\
      R_A^c&\leq\sum_{n=1}^N\rho_{n,2}^c R_{A,n,2}^c \\
      R_B^c&\leq\sum_{n=1}^N\rho_{n,2}^c R_{B,n,2}^c
    \end{split}
  \end{equation}
  where $R_{k,n,1}^c=C(p_{k,R,n}^c|h_{k,R,n}|^2)$, $R_{k,n,2}^c=C(p_{R,n}^c|h_{R,k',n}|^2)$, $k,k'\in\{A,B\},k\neq k'$, and $R_{AB,n,1}^c=C(p_{A,R,n}^c|h_{A,R,n}|^2+p_{B,R,n}^c|h_{B,R,n}|^2)$.

\end{proposition}

\begin{proof}
Please see Appendix~\ref{app:capacity}.
\end{proof}

Note that the capacity of OFDM two-way DF relaying derived in
\eqref{eqn:capacity} is  different from
the single-channel or per-subcarrier cases, since
\eqref{eqn:capacity} allows the relay to jointly decode and
re-encode the received signal from one set of subcarriers in the
first hop (MAC phase), and then forward the processed signal over a
different set of subcarriers in the second hop (BC phase).

Note that we focus on the achievable rate region of two-way DF relaying and assume that there exists an optimal coding/encoding approach to achieve the region. We further assume that channel coding is independently done at individual subcarriers, such that the frequency diversity can be exploited by  transmission mode selection in OFDM systems.

We now can characterize the achievable rate of
user $k$ over all the possible transmission modes:
\begin{equation}
  R_k =R_{k}^a+R_{k}^b+R_{k}^c,~k\in\{A,B\}.
\end{equation}

As shown in \cite{Caire}, delay-sensitive or delay-constrained transmission can be regarded as a delay-limited capacity problem, where a constant data rate should be maintained with probability one regardless of channel variations. Thus we consider constant data rates as the QoS requirements in this paper. Each user has its own rate QoS requirement, which can be expressed as
\begin{equation}\label{eqn:qos}
  R_{k}\geq r_{k},~k\in\{A,B\},
\end{equation}
where $r_{k}$ is the minimum rate requirements of user $k$.

Our objective is not only to optimally assign subcarriers
and transmission modes but also to allocate power and rate for each
user so as to maximize the weighted sum rates while maintaining
the individual rate requirements of each user. Mathematically, the joint optimization problem can be formulated as
(\textbf{P1})
\begin{subequations}
\begin{align}
\textbf{P1}:~\max_{\{\bs p, \bs\rho,\bs R\}}&\sum_{k\in\{A,B\}}w_kR_k\\
s.t.~&R_k^b\leq\sum_{n=1}^N\rho_{k,n,1}^b R_{k,n,1}^b \label{eqn:lamb1}\\
&R_k^b\leq\sum_{n=1}^N\rho_{k,n,2}^b R_{k,n,2}^b \label{eqn:lamb2}\\
&R_k^c\leq\sum_{n=1}^N\rho_{n,1}^c R_{k,n,1}^c \label{eqn:lamc1}\\
&R_k^c\leq\sum_{n=1}^N\rho_{n,2}^c R_{k,n,2}^c \label{eqn:lamc2}\\
&R_A^c+R_B^c\leq\sum_{n=1}^N\rho_{n,1}^cR_{AB,n,1}^c \label{eqn:lamc}\\
&\eqref{eqn:rho},\eqref{eqn:pk_peak},\eqref{eqn:pr_peak},\eqref{eqn:qos},\bs\rho\in\{0,1\},
\end{align}
\end{subequations}
where $w_k$ is the weight that represents the priority of user $k$,
$\bs p\triangleq\{p_{k,n}^a, p_{k,R,n}^b, p_{R,k,n}^b, p_{k,R,n}^c, p_{R,n}^c\}$ is the set of power variables, $\bs\rho\triangleq\{\rho_{k,n}^a,\rho_{k,n,i}^b,\rho_{n,i}^c\}$ is the set of assignment variables, and $\bs R\triangleq\{R_k^a,R_k^b , R_k^c\}$ is the set of rate variables.

Comparing with the related works \cite{Jitvanichphaibool,Jang,Dong,Ho,YuanTWC10,YuanTCOM12}, there are several unique features about our problem formulation \textbf{P1}. First, we jointly optimize subcarrier assignment, transmission mode selection, and power allocation. The previous works only consider partial resources of \textbf{P1}. Second, \textbf{P1} represents the first attempt that optimizes OFDM two-way DF relaying based on subcarrier set basis, according to the derived capacity region in \eqref{eqn:capacity}. Third, \textbf{P1} considers individual rate-QoS for each user, and we impose three transmission modes simultaneously to support the individual QoS but over different sets of subcarriers, thanks to the parallel structure of OFDM relaying.

\section{Optimal Transmission Policy}

The problem in \textbf{P1} is a mixed integer optimization problem. Finding the optimal solution needs exhaustive search with exponential complexity, where each subcarrier has eight possibilities of assignments over different users, different transmission modes, along with different hops. Total $N$ subcarriers are used and therefore total $8^N$ possibilities of assignments are needed. Then, power allocation is performed for each assignment (the pure power allocation problem is convex if the assignment is fixed), and the optimal solution follows the assignment that results in maximum throughput.
In this section, we present an efficient method to find the asymptotically optimal solution of \textbf{P1} with linear complexity in the number of subcarriers.

\subsection{Continuous Relaxation}
To make \textbf{P1} more tractable, we relax the binary variables $\bs\rho$ into real-valued ones, i.e., $\bs\rho\in[0,1]$. This continuous relaxation makes $\bs\rho$ as the time sharing factors for subcarriers. In addition, we introduce a set of new variables
$\bs s\triangleq\{p_{k,n}^a\rho_{k,n}^a, p_{k,R,n}^b\rho_{k,n,1}^b, p_{R,k,n}^b\rho_{k,n,2}^b, p_{k,R,n}^c\rho_{n,1}^c, p_{R,n}^c\rho_{n,2}^c\}$. Clearly, $\bs s$ can be viewed as the actual consumed powers on subcarriers.
Substituting $\bs s$ and real-valued $\bs\rho$ into the rate variables $\bs R$, the relaxed problem of \textbf{P1} then can be written as
\begin{subequations}
\begin{align}
\textbf{P2}:~\max_{\{\bs s, \bs\rho,\bs R\}}&\sum_{k\in\{A,B\}}w_kR_k\\
s.t.~&\sum_{n=1}^N\left( s_{k,n}^a + s_{k,R,n}^b + s_{k,R,n}^c \right)\leq
  P_k,~k\in\{A,B\}\label{eqn:s1}\\
  &\sum_{n=1}^N\left( \sum_{k\in\{A,B\}}s_{R,k,n}^b + s_{R,n}^c \right)\leq
  P_R\label{eqn:s2}\\
&\eqref{eqn:rho},\eqref{eqn:qos},\eqref{eqn:lamb1}-\eqref{eqn:lamc},
\bs\rho\in[0,1].
\end{align}
\end{subequations}

It is easy to find that each element of $\bs R$ has the form of $\rho\log_2(1+s/\rho)$ that is jointly concave in $(\rho,s)$, where $\rho$ and $s$ represent the general expressions of the elements in $\bs\rho$ and $\bs s$, respectively. Thus, the objective function of \textbf{P2} is concave since any
positive linear combination of concave functions is concave. Moreover, the constraints \eqref{eqn:s1}, \eqref{eqn:s2} and \eqref{eqn:rho} are affine, and the constraints \eqref{eqn:qos} and \eqref{eqn:lamb1}-\eqref{eqn:lamc} are convex.
Therefore, \textbf{P2} is a convex optimization problem.



We first introduce non-negative
Lagrangian multipliers $\lambda_k^{b_1}$, $\lambda_k^{b_2}$,
$\lambda_k^{c_1}$, $\lambda_k^{c_2}$, $\lambda_{AB}^c$ with
constraints \eqref{eqn:lamb1}-\eqref{eqn:lamc}, respectively. All of
them are denoted as $\bs\lambda\succeq0$. In addition,
non-negative Lagrangian multipliers
$\bs\alpha=\{\alpha_A,\alpha_B,\alpha_R\}\succeq0$ are introduced to
associate with the power constraints of the three nodes,
$\bs\mu=\{\mu_A,\mu_B\}\succeq0$ are associated the two users' QoS
requirements in \eqref{eqn:qos}. Then the dual function of
\textbf{P2} can be defined as
\begin{eqnarray}\label{eqn:dual}
g(\bs\lambda,\bs\alpha,\bs\mu)\triangleq
\max_{ \{\bs{s,\rho,R}\}\in\mathcal{D}}\mathcal L(\bs
s,\bs\rho,\bs R,\bs\lambda,\bs\alpha,\bs\mu),
\end{eqnarray}
where $\mathcal{D}$ is the set of all primal variables $\{\bs s,\bs\rho,\bs R\}$ that satisfy the constraints, and the Lagrangian is
\begin{eqnarray}\label{eqn:La}
\mathcal L(\bs s,\bs\rho,\bs R,\bs\lambda,\bs\alpha,\bs\mu)=\sum_{k\in\{A,B\}}w_kR_k+L^b+L^c+ \sum_{k\in\{A,B\}}\alpha_k\left[ P_k-\sum_{n=1}^N\left( s_{k,n}^a + s_{k,R,n}^b + s_{k,R,n}^c \right)\right]\nonumber\\
  + \alpha_R\left[P_R - \sum_{n=1}^N\left( \sum_{k\in\{A,B\}}s_{R,k,n}^b + s_{R,n}^c \right)\right]+\sum_{k\in\{A,B\}}\mu_k\left[(R_{k}^a+R_k^b+R_k^c)-r_k \right],
\end{eqnarray}
in which
\begin{eqnarray}\label{eqn:lb}
  L^b=\sum_{k\in\{A,B\}}\Bigg[\lambda_k^{b_1}\left(\sum_{n=1}^N \rho_{k,n,1}^bR_{k,n,1}^b-R_k^b\right)
  +\lambda_k^{b_2}\left(\sum_{n=1}^N \rho_{k,n,2}^bR_{k,n,2}^b-R_k^b\right)\Bigg],
\end{eqnarray}
\begin{eqnarray}\label{eqn:lc}
  L^c=\sum_{k\in\{A,B\}}\Bigg[\lambda_k^{c_1}\left(\sum_{n=1}^N \rho_{n,1}^cR_{k,n,1}^c-R_k^c\right)+\lambda_k^{c_2}\left(\sum_{n=1}^N \rho_{n,2}^cR_{k,n,2}^c-R_k^c\right)\Bigg]\nonumber\\
  +\lambda_{AB}^c\left(\sum_{n=1}^N \rho_{n,1}^cR_{AB,n,1}^c-R_A^c-R_B^c\right).
\end{eqnarray}

Computing the dual function $g(\bs\lambda,\bs\alpha,\bs\mu)$
requires to determine the optimal $\{\bs s,\bs\rho,\bs R\}$ for
given dual variables $\{\bs\lambda,\bs\alpha,\bs\mu\}$. In the
following we present the derivations in detail.

\subsection{Optimizing $\{\bs s,\bs\rho,\bs R\}$ for Given $\{\bs\lambda,\bs\alpha,\bs\mu\}$}

\subsubsection{Maximizing Lagrangian over $\bs R$}

Firstly, we look at the rate variables $\bs R$. It is seen that the optimal rates of direct transmission $\{R_k^a\}$ are exactly the capacity expressions, and the rates of the two hops for both one- and two-way relaying need to be coordinated (see \eqref{eqn:lb} and \eqref{eqn:lc}). Therefore we define a new rate set $\bs R'=\{R_k^b , R_k^c\}$ and the part of dual function with respect
to $\bs R'$ is given by
\begin{eqnarray}
  g_0(\bs\lambda,\bs\alpha,\bs\mu)=\max_{\bs R'}\sum_{k\in\{A,B\}}\Big[(w_k+\mu_k-\lambda_k^{b_1}-\lambda_k^{b_2})R_k^b
  +(w_k+\mu_k-\lambda_k^{c_1}-\lambda_k^{c_2}-\lambda_{AB}^c)R_k^c\Big].
\end{eqnarray}

To make sure the dual function is bounded, we have
$w_k+\mu_k-\lambda_k^{b_1}-\lambda_k^{b_2}=0$ and
$w_k+\mu_k-\lambda_k^{c_1}-\lambda_k^{c_2}-\lambda_{AB}^c=0$. In such case, $g_0(\bs\lambda,\bs\alpha,\bs\mu)\equiv0$ and we obtain that
\begin{eqnarray}
  &&\lambda_k^{b_2}=w_k+\mu_k-\lambda_k^{b_1},\label{eqn:lambb2}\\
  &&\lambda_k^{c_2}=w_k+\mu_k-\lambda_k^{c_1}-\lambda_{AB}^c.\label{eqn:lamcc2}
\end{eqnarray}

By substituting these results above into \eqref{eqn:La}, the Lagrangian
can be rewritten as:
\begin{eqnarray}\label{eqn:LaH}
  \mathcal L(\bs s,\bs\rho,\bs\lambda,\bs\alpha,\bs\mu)=  \sum_{n=1}^N\left[\sum_{k\in\{A,B\}}\left(H_{k,n}^a+H_{k,n}^{b_1}+H_{k,n}^{b_2} \right) +H_{n}^{c_1}+H_{n}^{c_2}\right] \nonumber\\
  +\sum_{k\in\{A,B,R\}}\alpha_k P_k - \sum_{k\in\{A,B\}}\mu_k r_k,
\end{eqnarray}
where
\begin{eqnarray}
&&  H_{k,n}^a=(w_k+\mu_k)\rho_{k,n}^aR_{k,n}^a-\alpha_k s_{k,n}^a,~k\in\{A,B\},\label{eqn:Ha}\\
&&  H_{k,n}^{b_1}=\lambda_k^{b_1}\rho_{k,n,1}^bR_{k,n,1}^b-\alpha_k s_{k,R,n}^b,~k\in\{A,B\},\label{eqn:Hb1}\\
&&  H_{k,n}^{b_2}=(w_k+\mu_k-\lambda_k^{b_1})\rho_{k,n,2}^bR_{k,n,2}^b-\alpha_R s_{R,k',n}^b,~k\in\{A,B\},\label{eqn:Hb2}\\
&&  H_{n}^{c_1}=\sum_{k\in\{A,B\}}\lambda_k^{c_1}\rho_{n,1}^cR_{k,n,1}^c+\lambda_{AB}^c \rho_{n,1}^cR_{AB,n,1}^c -\sum_{k\in\{A,B\}}\alpha_k s_{k,R,n}^c,\label{eqn:Hc1}\\
&&  H_{n}^{c_2}=\sum_{k\in\{A,B\}}(w_k+\mu_k-\lambda_k^{c_1}-\lambda_{AB}^c) \rho_{n,2}^cR_{k,n,2}^c-\alpha_R s_{R,n}^c.\label{eqn:Hc2}
\end{eqnarray}
For brevity, we denote $\xi_k=w_k+\mu_k-\lambda_k^{c_1}-\lambda_{AB}^c$ in what follows. As aforementioned, the two users should be both active for two-way relaying, such that the user index $k$ is not involved in $H_{n}^{c_1}$ and $H_{n}^{c_2}$.

Notice that the dual variables $\bs\mu$ and $\bs\alpha$ can be interpreted as QoS weights and power prices, respectively, then \eqref{eqn:Ha}-\eqref{eqn:Hc2} can be regarded as the \emph{profits} of different traffic sessions, which are defined as the QoS-aware throughput of traffic sessions minus the corresponding power costs. In what follows, we show that the profits defined in \eqref{eqn:Ha}-\eqref{eqn:Hc2} play a key role to derive $\{\bs s^*,\bs\rho^*\}$.

\subsubsection{Maximizing Lagrangian over $\bs s$}
Observing the Lagrangian in \eqref{eqn:LaH}, we find that the dual function in \eqref{eqn:dual} can be decomposed into $N$ independent functions with the identical structure:
\begin{equation}\label{eqn:dual-decomp}
  g(\bs\lambda,\bs\alpha,\bs\mu)=\sum_{n=1}^Ng_n(\bs\lambda,\bs\alpha,\bs\mu)+ \sum_{k=A,B,R}\alpha_k P_k - \sum_{k\in\{A,B\}}\mu_k r_k,
\end{equation}
where
\begin{eqnarray}\label{eqn:gn}
  g_n(\bs\lambda,\bs\alpha,\bs\mu)\triangleq\max_{\{\bs{s,\rho}\}\in\mathcal{D}}\mathcal L_n(\bs s,\bs\rho,\bs\lambda,\bs\alpha,\bs\mu)
\end{eqnarray}
with
\begin{eqnarray}\label{eqn:Ln}
  \mathcal L_n(\bs s,\bs\rho,\bs\lambda,\bs\alpha,\bs\mu)=\sum_{k\in\{A,B\}}\left(H_{k,n}^a+H_{k,n}^{b_1}+H_{k,n}^{b_2} \right) +H_{n}^{c_1}+H_{n}^{c_2}.
\end{eqnarray}
Note that the profits $H_{k,n}^a$, $H_{k,n}^{b_1}$, $H_{k,n}^{b_2}$, $H_{n}^{c_1}$, and $H_{n}^{c_2}$ in \eqref{eqn:Ln} are defined in \eqref{eqn:Ha}-\eqref{eqn:Hc2}, respectively.

We now solve $g_n(\bs\lambda,\bs\alpha,\bs\mu)$. Here we first analyze the optimal power allocations $\bs s^*$ for given subcarrier assignment and transmission mode selection $\bs\rho$.

 By applying Karush-Kuhn-Tucker (KKT) conditions \cite{Boyd}, the optimal power allocations for direct transmission are given by
\begin{equation}\label{eqn:pd}
  s_{k,n}^{a,*} = \rho_{k,n}^a\cdot p_{k,n}^{a,*}=\rho_{k,n}^a\cdot\left(\frac{w_k+\mu_k}{\sigma\alpha_k}-\frac{1}{|h_{k,k',n}|^2}\right)^+,
\end{equation}
with $k,k'\in\{A,B\},k\neq k'$, $\sigma\triangleq\ln2$ and $(x)^+\triangleq\max\{x,0\}$. \eqref{eqn:pd} shows that the optimal power
allocations for direct transmission are achieved by multi-level
water-filling. In particular, the water level of each user depends
explicitly on its QoS requirement and weight, and can
differ from one another.

By applying the KKT conditions, we obtain the
optimal power allocations for the first hop of one-way relaying:
\begin{equation}\label{eqn:pb1}
  s_{k,R,n}^{b,*}=\rho_{k,n,1}^b\cdot p_{k,R,n}^{b,*}= \rho_{k,n,1}^b\cdot\left(\frac{\lambda_k^{b_1}}{\sigma\alpha_k}-\frac{1}{|h_{k,R,n}|^2}
  \right)^+,~k\in\{A,B\}.
\end{equation}
Similarly, the optimal power allocations for the second hop of one-way relaying are given by
\begin{equation}\label{eqn:pb2}
  s_{R,k,n}^{b,*}= \rho_{k,n,2}^b\cdot p_{R,k,n}^{b,*}= \rho_{k,n,2}^b\cdot\left(\frac{w_k+\mu_k-\lambda_k^{b_1}}{\sigma\alpha_R}-\frac{1}{|h_{R,k',n}|^2}
  \right)^+
\end{equation}
with $k,k'\in\{A,B\},k\neq k'$. \eqref{eqn:pb1} and \eqref{eqn:pb2} show that the optimal power allocations for DF one-way relaying are also achieved by multi-level water-filling.

For the first hop (or MAC phase) of two-way relaying, the optimal power
allocation $s_{k,R,n}^{c,*}=\rho_{n,1}^c\cdot p_{k,R,n}^{c,*}$, where $p_{k,R,n}^{c,*}$ are the non-negative real root of the following equations:
\begin{eqnarray}\label{eqn:pmac}
\begin{cases}
  \frac{\lambda_A^{c_1}|h_{A,R,n}|^2}{1+p_{A,R,n}^c|h_{A,R,n}|^2}+
  \frac{\lambda_{AB}^c|h_{A,R,n}|^2}{1+p_{A,R,n}^c|h_{A,R,n}|^2+p_{B,R,n}^c|h_{B,R,n}|^2}=\sigma\alpha_A\\
  \frac{\lambda_B^{c_1}|h_{B,R,n}|^2}{1+p_{B,R,n}^c|h_{B,R,n}|^2}+
  \frac{\lambda_{AB}^c|h_{B,R,n}|^2}{1+p_{A,R,n}^c|h_{A,R,n}|^2+p_{B,R,n}^c|h_{B,R,n}|^2}=\sigma\alpha_B.
\end{cases}
\end{eqnarray}
%

%
%
%
It is readily found that \eqref{eqn:Hc2} is concave in $s_{R,n}^{c}$. Taking the partial derivative of \eqref{eqn:Hc2} with respect to $s_{R,n}^{c}$ and letting it be zero, the optimal power allocation for the second hop (or BC phase) of two-way relaying is  $s_{R,n}^{c,*}=\rho_{n,2}^{c}\cdot p_{R,n}^{c,*}$, where
\begin{eqnarray}\label{eqn:pbc}
  p_{R,n}^{c,*}=\begin{cases}0,&\textrm{if}~\alpha_R\geq\frac{\xi_B|h_{R,A,n}|^2+\xi_A|h_{R,B,n}|^2}{\sigma}
  \\\frac{-\phi_2+\sqrt{\phi_2^2-4\phi_1\phi_3}}{2\phi_1},&\textrm{otherwise},
\end{cases}
\end{eqnarray}
with $\phi_1=\alpha_R|h_{R,B,n}|^2|h_{R,A,n}|^2$,
$\phi_2=\alpha_R(|h_{R,B,n}|^2+|h_{R,A,n}|^2)-(\xi_A+\xi_B)|h_{R,B,n}|^2|h_{R,A,n}|^2/\sigma$,
and
$\phi_3=\alpha_R-(\xi_B|h_{R,A,n}|^2+\xi_A|h_{R,B,n}|^2)/\sigma$.

\subsubsection{Maximizing Lagrangian over $\bs\rho$}
Substituting the optimal power allocations $\bs
s^*(\bs\lambda,\bs\alpha,\bs\mu)$ into
\eqref{eqn:dual} to eliminate the power variables, the profits \eqref{eqn:Ha}-\eqref{eqn:Hc2} in the sub-Lagrangian \eqref{eqn:Ln} can be rewritten as respectively
\begin{eqnarray}
&&  H_{k,n}^a=\rho_{k,n}^a\left[(w_k+\mu_k)R_{k,n}^{a,*}-\alpha_k p_{k,n}^{a,*}\right],~k\in\{A,B\},\label{eqn:Ha'}\\
&&  H_{k,n}^{b_1}=\rho_{k,n,1}^b\left[\lambda_k^{b_1}R_{k,n,1}^{b,*}-\alpha_k p_{k,R,n}^{b,*}\right],~k\in\{A,B\},\label{eqn:Hb1'}\\
&&  H_{k,n}^{b_2}=\rho_{k,n,2}^b\left[(w_k+\mu_k-\lambda_k^{b_1})R_{k,n,2}^{b,*}-\alpha_R p_{R,k',n}^{b,*}\right],~k\in\{A,B\},\label{eqn:Hb2'}\\
&&  H_{n}^{c_1}=\rho_{n,1}^c\left[\sum_{k\in\{A,B\}}\lambda_k^{c_1}R_{k,n,1}^{c,*}+\lambda_{AB}^c R_{AB,n,1}^{c,*} -\sum_{k\in\{A,B\}}\alpha_k p_{k,R,n}^{c,*}\right],\label{eqn:Hc1'}\\
&&  H_{n}^{c_2}=\rho_{n,2}^c\left[\sum_{k\in\{A,B\}}(w_k+\mu_k-\lambda_k^{c_1}-\lambda_{AB}^{c,*}) R_{k,n,2}^{c,*}-\alpha_R p_{R,n}^{c,*}\right].\label{eqn:Hc2'}
\end{eqnarray}
Note that in \eqref{eqn:Ha'}-\eqref{eqn:Hc2'}, $\bs p^*$ are obtained from \eqref{eqn:pd}-\eqref{eqn:pbc} derived above, and then the optimal rates can also be computed correspondingly. Thus the profits \eqref{eqn:Ha'}-\eqref{eqn:Hc2'} are only related to the primal variables $\bs\rho$ for given dual variables $\{\bs\lambda,\bs\alpha,\bs\mu\}$.
Then the dual function over each subcarrier $n$ in \eqref{eqn:gn} can be rewritten as
\begin{eqnarray}\label{eqn:gn'}
g_n(\bs\lambda,\bs\alpha,\bs\mu)=\begin{cases}\max_{\bs\rho}& \mathcal L_n(\bs s^*,\bs\rho,\bs\lambda,\bs\alpha,\bs\mu)\\s.t.~&\eqref{eqn:rho},\bs\rho\in[0,1]. \end{cases}
\end{eqnarray}
The sub-Lagrangian $\mathcal L_n(\bs s^*,\bs\rho,\bs\lambda,\bs\alpha,\bs\mu)$ is defined in \eqref{eqn:Ln} with the profits \eqref{eqn:Ha'}-\eqref{eqn:Hc2'}, and also only related to the relaxed variables $\bs\rho$ for given dual variables. Now we are ready to find the optimal $\bs\rho^*$ based on the following proposition.

\begin{proposition}
There always exists an optimal \emph{binary} solution for $\bs\rho^*$ for the dual function \eqref{eqn:dual}.
\end{proposition}

\begin{proof}
For each subcarrier $n$, $\mathcal L_n(\bs s^*,\bs\rho,\bs\lambda,\bs\alpha,\bs\mu)$ has a bounded
objective and \eqref{eqn:gn'} is a linear programming over $\bs\rho_n\in[0,1]$, where  $\bs\rho_n\triangleq\{\rho_{A,n}^a,\rho_{B,n}^a,\rho_{A,n,1}^b,\rho_{B,n,1}^b,\rho_{A,n,2}^b,\rho_{B,n,2}^b,\rho_{n,1}^c,\rho_{n,2}^c\}$. A globally optimal solution can be found at the vertices
of the feasible region \cite{Garfinkel}. Therefore at least one optimal $\bs\rho_n^*$ is binary.
\end{proof}

According to Proposition 2 that at least one optimal $\bs\rho_n^*$ is binary, we resort to simple exhaustive search over all vertices for each subcarrier $n$, and follow the one that has the maximum value of $\mathcal L_n(\bs s^*,\bs\rho,\bs\lambda,\bs\alpha,\bs\mu)$ in \eqref{eqn:Ln}. Therefore the binary solution of $\bs\rho^*$ can be recovered. In other words, the optimal binary solution of $\bs\rho^*$ can be obtained as follows: Since $\mathcal L_n(\bs s^*,\bs\rho,\bs\lambda,\bs\alpha,\bs\mu)$ has eight profits defined in \eqref{eqn:Ha'}-\eqref{eqn:Hc2'}, each of them corresponds to one element of $\bs\rho_n$ ($\bs\rho_n$ has eight elements). Then, for each subcarrier $n$, by exhaustive search over all eight profits, let one out of the eight elements of $\bs\rho_n$  be 1 if its corresponding profit in $\mathcal L_n(\bs s^*,\bs\rho,\bs\lambda,\bs\alpha,\bs\mu)$ is maximum\footnote{Arbitrary tie-breaking can be performed if
necessary.} and others be 0.

It is also worth noting that the rates in the profits \eqref{eqn:Ha'}-\eqref{eqn:Hc2'} are the functions of channel state information (CSI) that are independent random variables. Thus the profits \eqref{eqn:Ha'}-\eqref{eqn:Hc2'} are also independent random variables. As a result, it is probability 0 that more than one profit have the same maximum value of $\mathcal L_n(\bs s^*,\bs\rho,\bs\lambda,\bs\alpha,\bs\mu)$.



\subsection{Optimizing Dual Variables $\{\bs\lambda,\bs\alpha,\bs\mu\}$}

After computing $g(\bs\lambda,\bs\alpha,\bs\mu)$, we
now solve the standard dual optimization problem which is
\begin{eqnarray}\label{eqn:min}
  \min_{\bs\lambda,\bs\alpha,\bs\mu}&&g(\bs\lambda,\bs\alpha,\bs\mu)\label{eqn:min-obj}\\
  s.t.~~&&-\bs\lambda,-\bs\alpha,-\bs\mu \preccurlyeq 0,\label{eqn:min1}\\
  &&-w_k-\mu_k+\lambda_k^{b_1} \leq 0, k \in \{A,B\}\label{eqn:min2} \\
  &&-w_k-\mu_k+\lambda_k^{c_1}+\lambda_{AB}^c \leq 0, k \in \{A,B\}.\label{eqn:min3}
\end{eqnarray}

Since a dual function is always convex by definition, the commonly used gradient based algorithms or ellipsoid method can be employed to update $\{\bs\lambda,\bs\alpha,\bs\mu\}$ toward optimal $\{\bs\lambda^*,\bs\alpha^*,\bs\mu^*\}$ with global convergence  \cite{Boyd}. In this paper we use ellipsoid method to update $\{\bs\lambda,\bs\alpha,\bs\mu\}$ simultaneously based on the following proposition.
\begin{proposition}
For the dual problem \eqref{eqn:min}, the subgradient vector is
\begin{eqnarray}\label{eqn:delta}
\bs\Delta=\left[\begin{array}{l}
    \Delta\lambda_A^{b_1} = {\sum_{n=1}^N (R_{A,n,1}^b-R_{A,n,2}^b)} \\
    \Delta\lambda_B^{b_1} = {\sum_{n=1}^N (R_{B,n,1}^b-R_{B,n,2}^b)} \\
    \Delta\lambda_A^{c_1} = {\sum_{n=1}^N (R_{A,n,1}^c-R_{A,n,2}^c)} \\
    \Delta\lambda_B^{c_1} = {\sum_{n=1}^N (R_{B,n,1}^c-R_{B,n,2}^c)} \\
    \Delta\lambda_{AB}^c = {\sum_{n=1}^N( R_{AB,n,1}^c - R_{A,n,2}^c -R_{B,n,2}^c)} \\
    \Delta\mu_A = {R_A^a+\sum_{n=1}^N\left(R_{A,n,2}^b+R_{A,n,2}^c\right) - r_A} \\
    \Delta\mu_B = {R_B^a+\sum_{n=1}^N\left(R_{B,n,2}^b+R_{B,n,2}^c\right) - r_B} \\
    \Delta\alpha_A = {P_A-\sum_{n=1}^N \left(s_{A,n}^a+s_{A,R,n}^b+s_{A,R,n}^c\right)} \\
    \Delta\alpha_B = {P_B-\sum_{n=1}^N \left(s_{B,n}^a+s_{B,R,n}^b+s_{B,R,n}^c\right)} \\
    \Delta\alpha_R = {P_R-\sum_{n=1}^N \left(\sum_{k\in\{A,B\}}s_{R,k,n}^b+s_{R,n}^c\right)} \\
\end{array}
\right]
\end{eqnarray}

\end{proposition}


\subsection{Discussions on Optimality and Complexity}

It is worth noting that, given any $\{\bs\lambda,\bs\alpha,\bs\mu\}$, there may exist non-integer optimal solutions for maximizing $\mathcal L_n(\bs\rho,\bs\lambda,\bs\alpha,\bs\mu)$ in \eqref{eqn:Ln}. In this case, more than one profit have the maximum value among the eight profits in $\mathcal L_n(\bs\rho,\bs\lambda,\bs\alpha,\bs\mu)$. As stated in Proposition 2, we choose only one of the optimal solutions in binary form to satisfy the primal exclusive subcarrier assignment constraints.

We also note that, for the subcarriers whose $\mathcal L_n(\bs\rho,\bs\lambda,\bs\alpha,\bs\mu)$ has multiple maximum profits, the binary subcarrier assignments may not be feasible for the primal power constraint(s). The key point is that the Lagrangian may not be differentiable at some given $\{\bs\lambda,\bs\alpha,\bs\mu\}$ (but the subgradients exist). Thus, the small variation of the dual variables $\{\bs\lambda,\bs\alpha,\bs\mu\}$ may change the binary assignment variables $\bs\rho$, and then result in a quantum leap on the sum power(s). In this case, though the dual variables $\{\bs\lambda,\bs\alpha,\bs\mu\}$ converge to an optimum, the allocated powers may exceed the primal power constraint(s). In other words, the ``duality gap" exists.
However, as shown in \cite{Yu,Ng}, the duality gap becomes zero under the so-called ``time-sharing" condition, and the time-sharing condition is always satisfied as the number of subcarriers increases in multicarrier systems. Then the global optimum can be obtained accurately in dual domain.
Briefly, as the argument in \cite{Yu} and \cite{Ng}, if two sets of rates using two different transmission policies are achievable individually, then their linear combination is also achievable by a frequency-division multiplex of the two transmission policies. This is possible when the number of subcarriers goes to large, the channel gains of adjacent subcarriers become more and more similar to each other. As a result, the same performance as that of time-sharing can be achieved by frequency-sharing without implementing the actual time-sharing.
%

Finally, we summarize the proposed dual-based solution in Algorithm 1.
Note that the dual problem in \eqref{eqn:min-obj}-\eqref{eqn:min3} is a standard inequality constrained problem. For such a problem, the ellipsoid update depends on whether the inequality constraints \eqref{eqn:min-obj}-\eqref{eqn:min3} are met. That is, if the dual variables are feasible (i.e., the inequality constraints are met), the subgradients are chosen as the unconstrained case (i.e., $\bs\Delta$), and otherwise the subgradients are chosen as the subgradients of the constraints. The detailed update rule can be found in \cite{Ellipsoid}.
%
%
In Algorithm 1, for given transmit powers, the
system is said to be in an outage if any QoS rate
requirement can not be satisfied. In this case, we set the rates as zero.
The computational complexity of the
ellipsoid method is $\mathcal O(q^2)$, where $q$ is the number of
the dual variables and $q=10$ in our case. Combining the complexity of
decomposition in \eqref{eqn:dual-decomp}, the total complexity of
the proposed algorithm is $\mathcal O(q^2N)$, which is linear in the
number of subcarriers.

\begin{algorithm}[tb]
\caption{Proposed dual-based method for \textbf{P1}}
\begin{algorithmic}[1]
\STATE \textbf{initialize} $\{\bs\lambda,\bs\alpha,\bs\mu\}$.
\REPEAT \STATE Compute the profits
$\{H_{k,n}^a,H_{k,n}^{b_1},H_{k,n}^{b_2},H_{n}^{c_1},H_{n}^{c_2}\}$
using the optimal power allocations $\bs
s^*(\bs\lambda,\bs\alpha,\bs\mu)$ derived in
\eqref{eqn:pd}-\eqref{eqn:pbc} for all $k$ and $n$.  \STATE Compare the profits
for each subcarrier $n$, and let the maximum profit be active and others
inactive.  Then the optimal
$\bs\rho^*(\bs\lambda,\bs\alpha,\bs\mu)$ can be obtained. \STATE
Update $\{\bs\lambda,\bs\alpha,\bs\mu\}$ using the ellipsoid method as the following steps 6-10:
 \IF{the constraints \eqref{eqn:min1}-\eqref{eqn:min3} are all satisfied} \STATE Update the ellipsoid with $\bs\Delta$ defined in \eqref{eqn:delta}. \ELSE \STATE
 Update the ellipsoid with the gradient of the constraints \eqref{eqn:min1}-\eqref{eqn:min3}.\ENDIF
 \UNTIL{$\{\bs\lambda,\bs\alpha,\bs\mu\}$ converge.}\IF{$R_A\geq r_A$ and $R_B\geq r_B$}
 \STATE QoS requirements are satisfied and output  $R_A$ and $R_B$. \ELSE \STATE Declare an outage and output $R_A=R_B=0$.\ENDIF
\end{algorithmic}

\end{algorithm}

\begin{remark}
Note that we consider the classical three-node bidirectional transmission model only for obtaining more insights and ease of presentation. The proposed optimization framework and algorithm can be extended to general multi-pair multi-relay scenario. Briefly, if there are $K$ user pairs and $M$ relays, by solving the dual problem of the original problem (the details are omitted here), the optimal power allocations have the same structures as \eqref{eqn:pd}-\eqref{eqn:pbc}. Then, for each subcarrier $n$ we obtain: for direct transmission, $2K$ profits all having the same structure as \eqref{eqn:Ha} and each for one user; similarly, for the first (or second) hop of one-way relaying, $2KM$ profits all having the same structure as \eqref{eqn:Hb1} (or \eqref{eqn:Hb2}) and each for one user-relay pair; for the first (or second) hop of two-way relaying, $KM$ profits all having the same structure as \eqref{eqn:Hc1} (or \eqref{eqn:Hc2}) and each for one relay and one user pair. According to the idea in Proposition 2, we assign each subcarrier $n$ to the traffic that has the maximum profit among the total $2K+6KM$ profits. Finally, the gradient or ellipsoid method can be used to find the optimal dual variables with polynomial complexity.

\end{remark}

\section{Simulation Results}

In this section, we conduct comprehensive simulation to evaluate the
performance of the proposed scheme.  The performance of two
benchmarks, namely BM1 and BM2, are presented. In BM1, the two users transmit directly without the assistance of the relay. Compared with the proposed scheme, BM2 has no two-way relaying transmission.
Note that these two benchmarks are the special cases of the proposed scheme and can be solved by the proposed algorithm with complexity $\mathcal O(q_1^2N)$ and $\mathcal O(q_2^2N)$, where $q_1=4$ and $q_2=7$ for BM1 and BM2, respectively. Note again that the complexity of the proposed algorithm is $\mathcal O(q^2N)$ with $q=10$, which is slightly higher than BM1 and BM2 but has the same order of complexity.
For brevity, we use DT, OW, and TW to denote direct transmission, one-
and two-way relaying in the simulation figures.

We set the distance between users $A$ and $B$ as $2$km, and the relay
$R$ is located in a line between the two users. The
Stanford University Interim (SUI)-6 channel model \cite{SUI} is employed to
generate OFDM channels and the path-loss exponent is fixed as $3.5$.
The number of subcarriers is set as $N=256$. Without loss of
generality, we let the three nodes have the same peak power
constraints (i.e., $P_A=P_B=P_R$) in dB. In all simulations, the stopping condition of the ellipsoid method (the details can be found in \cite{Boyd,Ellipsoid}) is set to be $10^{-4}$, which is accurate enough to the global optimum.

\subsection{Symmetric QoS Requirements and Relay Location}
In this subsection, we let the relay locate at the midpoint of the two users and $w_A=w_B=1$, the two users have the same rate requirements.

To clearly show the benefits of the proposed subcarrier set relaying, we plot the performance comparison between the capacity region derived \eqref{eqn:capacity} and conventional subcarrier pairing relaying in Fig. \ref{fig:gain}. For an illustration purpose, we assume equal power allocation, $N=8$ subcarriers, $w_A=w_B=1$, and $r_A=r_B=0$. For both two relaying methods, we adopt exhaustive search to find the optimal solutions. From the figure, we observe that the proposed two-way DF capacity region derived in \eqref{eqn:capacity} remarkably outperforms the conventional subcarrier pairing relaying, e.g., about $35\%$ throughput gain can be achieved when SNR=$20$dB.

We compare the system throughput performance of the proposed scheme and
the two benchmarks in Fig.~\ref{fig:rate}, where $r_A=r_B=5$bits/OFDM symbol.  It is observed that the
proposed algorithm significantly outperforms the benchmarks, which
clearly demonstrate the superiority of the proposed algorithm. For example, when signal-to-noise ratio (SNR) is $20$dB, the proposed scheme can achieve about $60\%$ and $10\%$ throughput improvements compared with BM1 and BM2. Moreover, the throughput improvements are increasing with SNR.

We then plot the outage performance of the three schemes in Fig. \ref{fig:outage}, where the QoS rate requirements $r_A=r_B=50$bits/OFDM symbol and $r_A=r_B=100$bits/OFDM symbol are considered. The system is said to be in an outage if any QoS rate requirement of the two users can not be satisfied. Compared with the two benchmarks, we observe that the proposed scheme can more efficiently support the QoS rate requirements.

Fig.~\ref{fig:subcnum} illustrates the number of occupied subcarriers
by different transmission modes, where $r_A=r_B=5$bits/OFDM symbol.  One can
observe that in low SNR regime (e.g.,
$10$dB), the three schemes do not occupy all subcarriers. This is
because in low SNR regime, no power is allocated to those
subcarriers with poor channel conditions. When SNR is high (e.g.,
$30$dB), we
observe that BM2 (DT together with OW) and the proposed scheme
occupy almost all subcarriers. However, some subcarriers are still discarded in
BM1 even when SNR is $30$dB. These observations show the benefits of
cooperative transmission. Finally, we find that the utilized
subcarriers for direct transmission are increasing with SNR in three
schemes, and the utilized subcarriers for two-way relaying are
increasing with SNR in our proposed scheme. Nevertheless, in both
BM2 and proposed schemes, the utilized subcarriers for one-way
relaying are increasing when SNR is less than $20$dB, and decreasing
when SNR is larger than about $20$dB.

Fig.~\ref{fig:ratepercen} shows the throughput percentages by different transmission modes
for BM2 and the proposed scheme, where $r_A=r_B=5$bits/OFDM symbol.  One observes that the throughput percentages of direct transmission and two-way relaying are increasing with SNR, but the importance of one-way relaying is decreasing as SNR increases. Moreover, in our proposed scheme, two-way relaying dominates the throughput performance. This suggests the significance of two-way relaying in the system.

\subsection{Asymmetric QoS Requirements and Relay Location}

In this subsection, we evaluate the performance of the proposed scheme when the two users' rate QoS requirements are asymmetric. The effects of relay location is also investigated. In this subsection we fix the transmit peak powers $P_A=P_B=P_R=20$dB.

Here we first let the relay node locate at the midpoint of user $A$ and user $B$. Fig.~\ref{fig:QoS_Ra} and Fig.~\ref{fig:QoS_Rb} show the throughput by different transmission modes of user $A$ and user $B$ versus different QoS requirements respectively, where $r_A+r_B=100$ bits/OFDM symbol. Fig.~\ref{fig:QoS_biased_20db} shows that two-way relaying contributes the highest throughput for both two users whatever the QoS requirements $r_A$ and $r_B$ vary. This is because two-way relaying must occur in pair, i.e., when two-way relaying generates throughput for user $A$, it also generates throughput for user $B$. We also find that when a user's rate requirement becomes more stringent, the effect of one-way relaying becomes more important for this user, and the effect of direct transmission becomes small.

We further consider the impacts of relay location in Fig.~\ref{fig:relay_position_20db}, where the relay node moves from user $A$ to user $B$ in a line. In this figure, the two users' rate requirements are fixed as $r_A=r_B=5$bits/OFDM symbol. It observes that direct transmission dominates the system performance when the relay node is close to either of the two users. This is because the fading channels between the relay and the distant user becomes the major limit of cooperative transmission (including one- and two-way relaying), which makes the cooperative transmission hardly happen. Moreover, both one- and two-way relaying perform their best when the relay node is at the midpoint of the two users.

\section{Conclusion}

In this paper, we studied the joint optimization problem of power
allocation, subcarrier assignment,  and transmission mode
selection with QoS guarantees in OFDM-based bidirectional
transmission systems. By using the dual method, we efficiently
solved the mix integer programming problem in an asymptotically
optimal manner. We also derived the capacity region of two-way DF
relaying in OFDM channels. Simulation results showed that our
proposed scheme can outperform the traditional schemes by a
significant margin.

A few interesting conclusions have been obtained through simulations.  First, the significance of one-way relaying is
decreasing with SNR. Second, the throughput percentages of direct
transmission and two-way relaying are increasing with SNR, and
two-way relaying dominates the system performance. Third, for a given user, one-way relaying contributes more throughput with the increasing rate requirements, and direct transmission performs oppositely. Finally, direct transmission dominates the system performance when the relay is closer either of the two users, and one- and two-way relay work well when the relay locates in the midpoint of the two users.

\appendices
\section{Proof of Proposition 1}\label{app:capacity}
The proposition can be proved by the similar way as \cite{Liang}, where the achievable capacity for traditional one-way relaying in parallel relay channels was derived.
Specifically, we first denote $\mathcal N_1 = \{n|\rho_{n,1}^c=1\}$ and $\mathcal N_2 = \{n|\rho_{n,2}^c=1\}$.
In the first hop, the received signals at the relay $R$ is given by
\begin{eqnarray}
  {Y_{R,n}}=\sqrt {p_{A,R,n}^c} {h_{A,R,n}}{X_{A,n}} + \sqrt {p_{B,R,n}^c} {h_{B,R,n}}{X_{B,n}}+ {Z_{R,n}} ,~ n \in \mathcal N_1.
\end{eqnarray}
In the second hop, the received signals at users $A$ and $B$  are
\begin{eqnarray}
  {Y_{A,n}}=\sqrt {p_{R,n}^c} {h_{R,A,n}}{X_{R,n}} + {Z_{A,n}} , n \in \mathcal N_2 ,\\
  {Y_{B,n}}=\sqrt {p_{R,n}^c} {h_{R,B,n}}{X_{R,n}} + {Z_{B,n}} , n \in  \mathcal N_2 ,
\end{eqnarray}
respectively.

Denote ${X_k} = \{ X_{k,n} | n \in\mathcal N_1 \},~k \in \{A,B \}$ ,
${Y_R} = \{ Y_{k,n} | n \in\mathcal N_1\}$ ,
${X_R} = \{ Y_{k,n} | n \in\mathcal N_2 \}$ and
${Y_k} = \{ Y_{k,n} | n \in\mathcal N_2\},~k \in \{ A, B \}$.
The achievable capacity region of this channel is the convex hull of all $\left(R_A^c, R_B^c\right)$ satisfying
\begin{eqnarray}
  R_A^c &\leqslant& \min \{ I\left( {{X_A};{Y_R}|{X_B}} \right),I\left( {{X_R};{Y_B}} \right)\},   \\
  R_B^c &\leqslant& \min \{ I\left( {{X_B};{Y_R}|{X_A}} \right),I\left( {{X_R};{Y_A}} \right)\},   \\
  R_A^c + R_B^c &\leqslant& I\left( {{X_A},{X_B};{Y_R}} \right).
\end{eqnarray}

Let the input signals $X_{A,n}$, $X_{B,n}$, and $X_{R,n}$ for each subcarrier be independent Gaussian
distributed, we obtain the achievable capacity region \eqref{eqn:capacity}. This completes the proof.

\bibliographystyle{IEEEtran}
\bibliography{IEEEabrv,Paper-TW-May-12-0709-df}

\begin{table}[b]
 \centering\small
 \begin{threeparttable}
 \caption{\label{tab:compare}Related Works Compared with This Paper}
  \begin{tabular}{cccccc}
  \toprule
  & basis & power allocation  & mode selection & subcarrier assignment & QoS
  \\
  \midrule
  Our paper & subcarrier-set & $\surd$ & $\surd$ & $\surd$ & $\surd$  \\\hline
  \cite{Jitvanichphaibool} & per-subcarrier & $\surd$ & $\times$ & $\surd$ & $\times$ \\\hline
  \cite{Jang} & per-subcarrier & $\surd$ & $\times$ & $\times$ & $\times$ \\\hline
  \cite{Dong} & per-subcarrier & $\surd$ & $\times$ & $\times$ & $\times$ \\\hline
  \cite{Ho} & subcarrier-pairing & $\surd$ & $\times$ & $\times$ & $\times$ \\\hline
  \cite{YuanTWC10} & subcarrier-pairing & $\times$ & $\surd$ & $\surd$ & $\times$ \\\hline
  \cite{YuanTCOM12} & subcarrier-pairing & $\times$ & $\times$ & $\surd$ & $\times$ \\
  \bottomrule
  \end{tabular}
  \end{threeparttable}
\end{table}

\begin{figure}[b]
\begin{centering}
\includegraphics[scale=1]{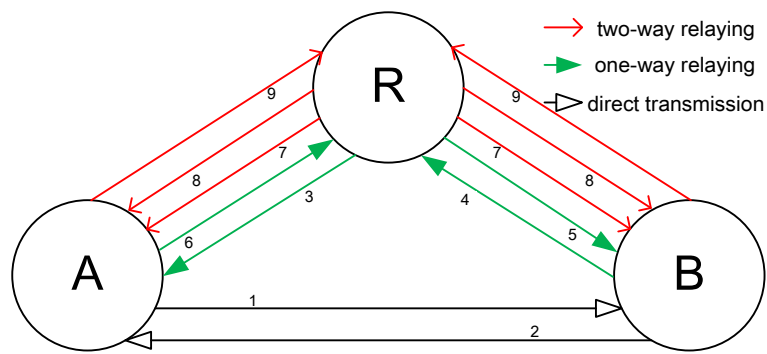}
\vspace{-0.1cm}
 \caption{Relay-assisted bidirectional transmission model, where the numbers denote the subcarrier indexes.}\label{fig:system}
\end{centering}
\vspace{-0.3cm}
\end{figure}
\begin{figure}[b]
\begin{centering}
\includegraphics[scale=.8]{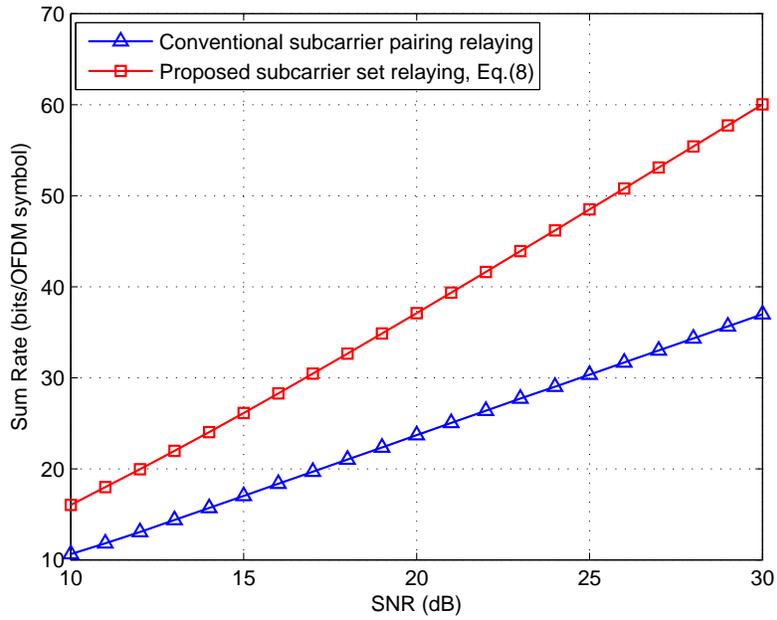}
\vspace{-0.1cm}
 \caption{Sum-rate comparison of the proposed subcarrier set relaying and the conventional subcarrier pairing relaying with $N=8$ subcarriers and equal power allocation, where $w_A=w_B=1$.}\label{fig:gain}
\end{centering}
\vspace{-0.3cm}
\end{figure}
\begin{figure}[b]
\begin{centering}
\includegraphics[scale=.8]{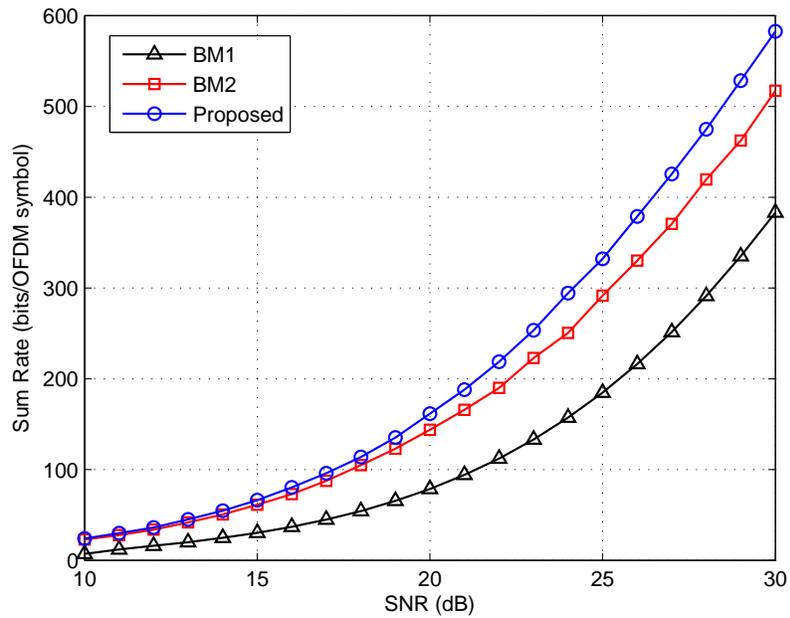}
\vspace{-0.1cm}
 \caption{Sum-rate performance of different schemes with $w_A=w_B=1$ and $r_A=r_B=5$bits/OFDM symbol.}\label{fig:rate}
\end{centering}
\vspace{-0.3cm}
\end{figure}
\begin{figure}[b]
\begin{centering}
\includegraphics[scale=.8]{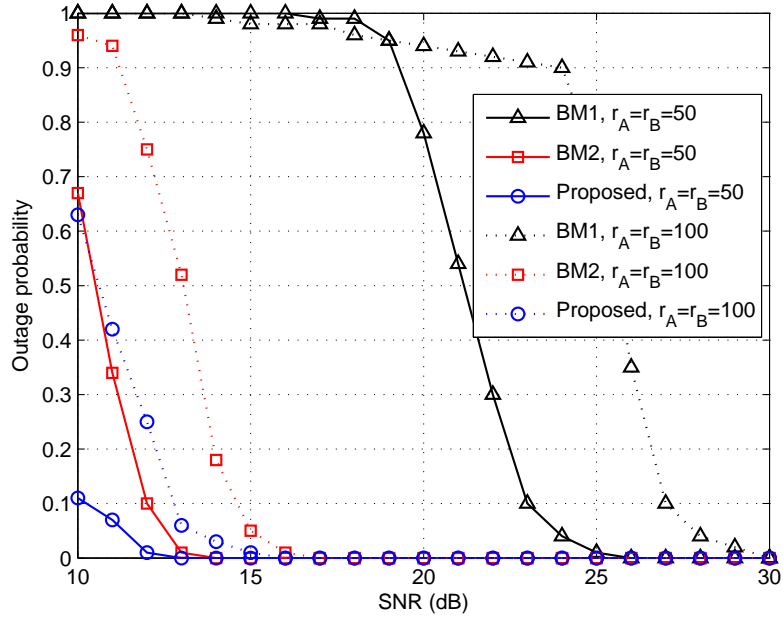}
\vspace{-0.1cm}
 \caption{Outage performance of different schemes with $r_A=r_B=50$bits/OFDM symbol and $r_A=r_B=100$bits/OFDM symbol, where $w_A=w_B=1$.}\label{fig:outage}
\end{centering}
\vspace{-0.3cm}
\end{figure}
\begin{figure}[b]
\begin{centering}
\includegraphics[scale=.8]{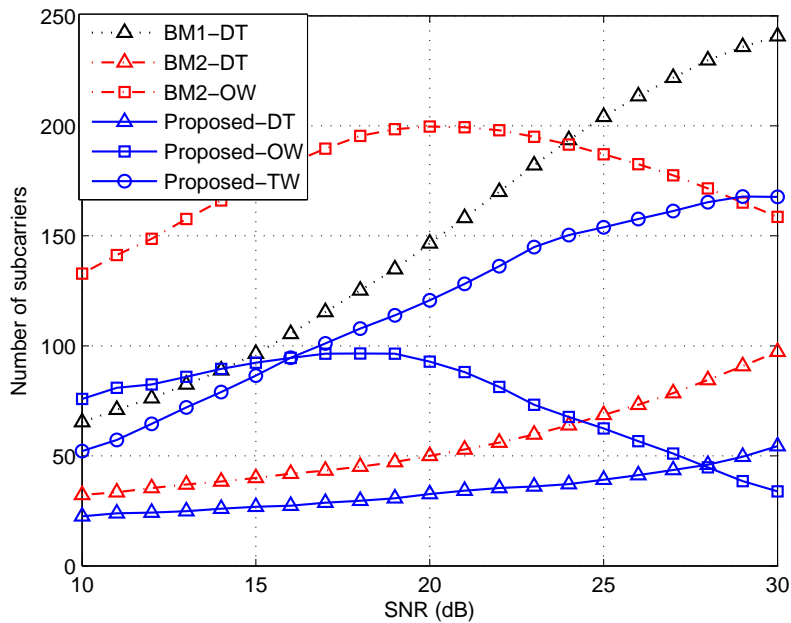}
\vspace{-0.1cm}
 \caption{Number of occupied subcarriers of difference transmission modes for different schemes with $w_A=w_B=1$ and $r_A=r_B=5$bits/OFDM symbol.}\label{fig:subcnum}
\end{centering}
\vspace{-0.3cm}
\end{figure}
%


%
\begin{figure}[t]
\begin{centering}
\makeatletter\def\@captype{figure}\makeatother
\subfigure[BM2]{\includegraphics[width=5in]{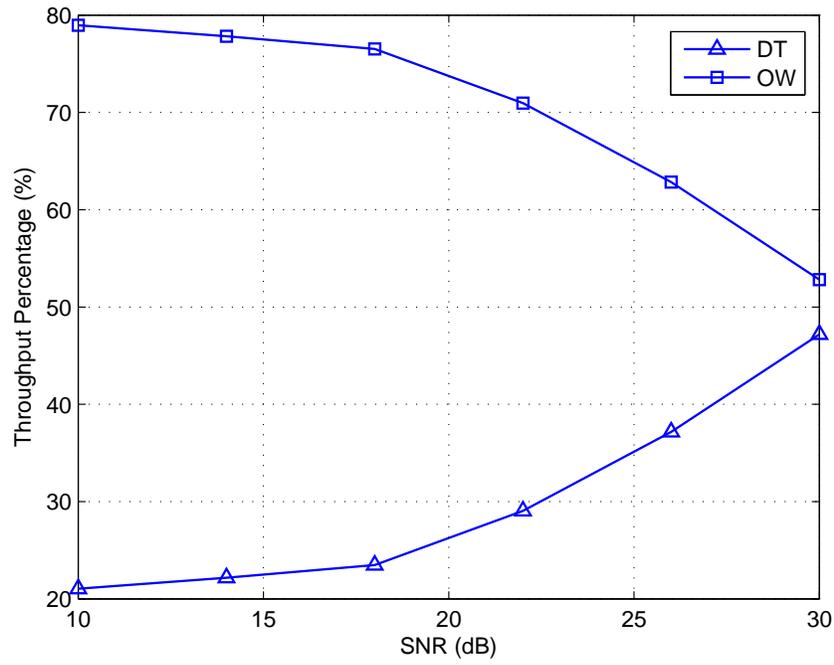}
\label{fig:percentageBM2}}
\subfigure[Proposed]{\includegraphics[width=5in]{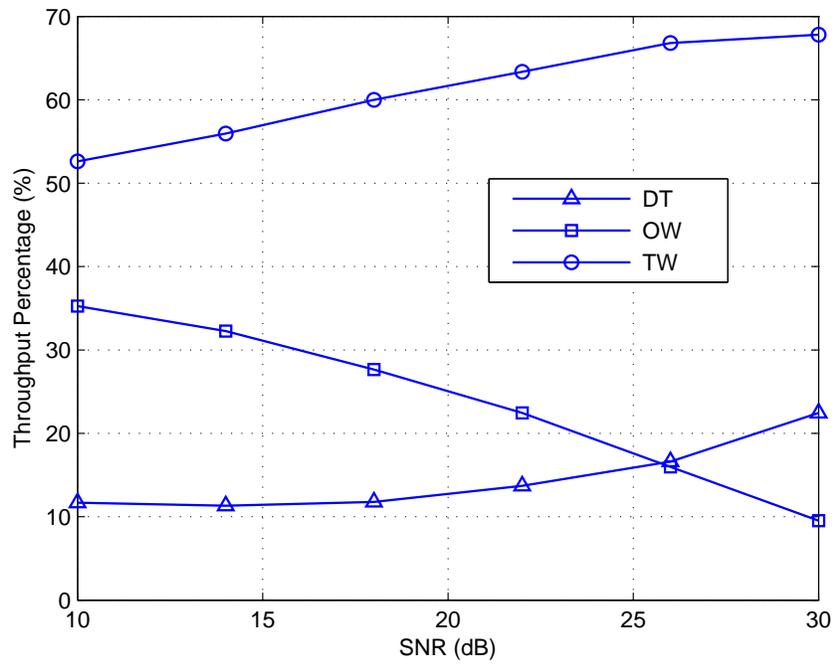}
\label{fig:percentageProposed}} \vspace{-0.1cm} \caption{Throughput percentages by different transmission modes with $w_A=w_B=1$ and $r_A=r_B=5$bits/OFDM symbol.} \label{fig:ratepercen}
\end{centering}
\vspace{-0.3cm}
\end{figure}
\begin{figure}[t]
\begin{centering}
\makeatletter\def\@captype{figure}\makeatother
\subfigure[$R_A$]{\includegraphics[width=5in]{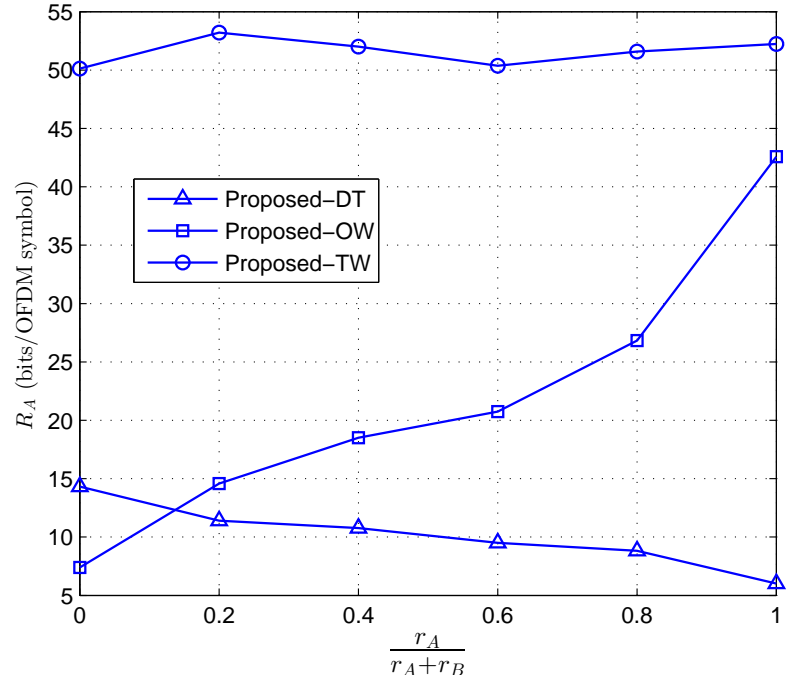}
\label{fig:QoS_Ra}}
\subfigure[$R_B$]{\includegraphics[width=5in]{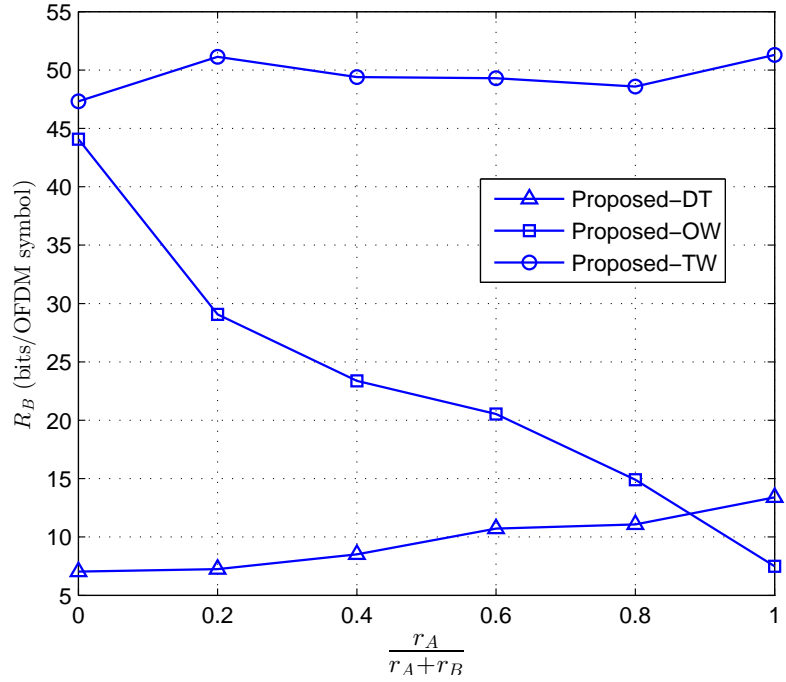}
\label{fig:QoS_Rb}} \vspace{-0.1cm} \caption{$R_A$ and $R_B$ versus
different QoS ratio $\frac{r_A}{r_A+r_B}$, where $r_A+r_B=100$ bits/OFDM symbol and the transmit power
is fixed as $20$dB. } \label{fig:QoS_biased_20db}
\end{centering}
\vspace{-0.3cm}
\end{figure}
\begin{figure}[b]
\begin{centering}
\includegraphics[scale=.8]{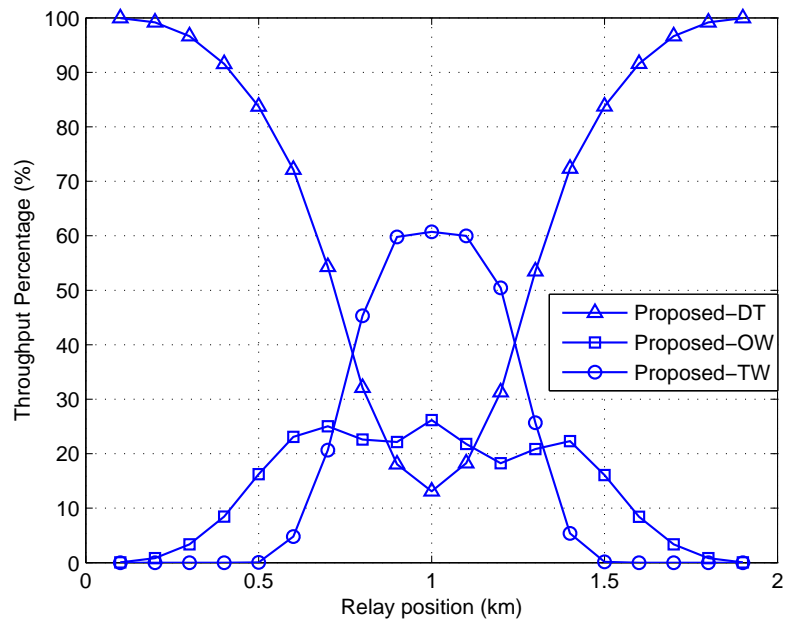}
\vspace{-0.1cm}
 \caption{Throughput percentages of different transmission modes, where the transmit power
is fixed as $20$dB and $r_A=r_B=5$bits/OFDM symbol.}\label{fig:relay_position_20db}
\end{centering}
\vspace{-0.3cm}
\end{figure}

\end{document}